\documentclass[letterpaper, 10 pt, conference]{ieeeconf}  

\IEEEoverridecommandlockouts 

\usepackage{afterpage}
\usepackage{epsfig} 
\usepackage{times} 
\usepackage{amsmath} 
\usepackage{amssymb}  
\usepackage{amsfonts}
\usepackage{mathptmx} 
\usepackage{tabularx}
\usepackage{epstopdf}
\usepackage{float}
\usepackage{epic,color}
\usepackage{mathrsfs}
\usepackage{dsfont,MnSymbol}
\usepackage{algorithm}
\usepackage{multirow} 
\usepackage[T1]{fontenc}

\newcounter{rmnum}

\newcounter{anum}

\def\IEEEQEDclosed{\mbox{\rule[0pt]{1.3ex}{1.3ex}}}

\def\qed{\ifmmode\IEEEQEDclosed\else{\unskip\nobreak\hfil
		\penalty50\hskip1em\null\nobreak\hfil\IEEEQEDclosed
		\parfillskip=0pt\finalhyphendemerits=0\endgraf}\fi}

\def\qed{\hspace*{\fill}~\IEEEQED\par\endtrivlist\unskip}

\def\Re{\mathbb{R}}

\def\notes#1{\marginpar{\tiny #1}\typeout{Notes!
Notes!
Notes!
}}
\renewcommand{\notes}[1]{\typeout{notes!}}

\def\Re{\field{R}}

\def\E{{\sf E}}

\def\IEEEQEDclosed{\mbox{\rule[0pt]{1.3ex}{1.3ex}}}
\def\qed{\nobreak\hfill\IEEEQEDclosed}

\def\beq{\begin{eqnarray}} 
\def\bc{\begin{center}} 
\def\be{\begin{enumerate}}
\def\bi{\begin{itemize}} 
\def\bs{\begin{small}}
\def\bS{\begin{slide}}
\def\ec{\end{center}} 
\def\ee{\end{enumerate}}
\def\ei{\end{itemize}}
\def\es{\end{small}}
\def\eS{\end{slide}}
\def\eeq{\end{eqnarray}}

\newcommand{\newP}[1]{\medskip\noindent{\bf #1:}}

\newcommand{\ud}{\,\mathrm{d}}

\def\Re{\mathbb{R}}
\def\E{{\sf E}}

\renewcommand{\Re}{\mathbb{R}}

\newcommand{\var}{\text{Var}}

\def\E{{\sf E}}
\def\bS{\mathbb{S}}

\def\tp{{\hbox{\rm\tiny T}}}

\def\opt{{\text{\rm (opt)}}}

\def\Re{\mathbb{R}}
\def\ud{\,\mathrm{d}}

\newtheorem{theorem}{Theorem}

\newtheorem{remark}{Remark}
\newtheorem{proposition}{Proposition}

\newlength{\noteWidth}
\long\def\notes#1{\ifinner
	{\tiny #1}
	\else
	\marginpar{\parbox[t]{\noteWidth}{\raggedright\tiny #1}}
	\fi}

\title{\LARGE \bf A Unified Control Theory Derivation of Discrete-Time Linear Ensemble Kalman Filters}

\author{Jin Won Kim
	\thanks{J. W. Kim is with the Department of Mechanical and System Design Engineering at the Hongik University. (e-mail: jin.won.kim@hongik.ac.kr)}
}

\begin{document}

\maketitle
\thispagestyle{empty}
\pagestyle{empty}

\begin{abstract}
	
The ensemble Kalman filter (EnKF) has become a standard methodology for state estimation in high-dimensional systems, yet its various stochastic and deterministic formulations often appear conceptually disconnected. In this paper, a unified derivation framework for EnKF algorithms are established by leveraging the classical duality between estimation and optimal control, which is the key concept in deriving Kalman filter. By extending the minimum variance estimation problem to cover the second order moment of the ensembles, we demonstrate that seemingly distinct EnKF variants---both with or without perturbed observation---can be systematically classified.

Specifically, the duality based framework reveals that the operational differences among these variety of EnKF algorithms reduce to a specific choice of hyperparameters. Ultimately, this perspective not only covers existing EnKF variants but also provides a systematic foundation for designing novel hybrid filters using control theory approach.

\end{abstract}

\section{Introduction}
\label{sec:intro}

The classical Kalman filter is coeval with the dawn of modern control theory. The fundamental duality between control and estimation has been noted since the early 1960s for both linear~\cite{kalman1960general} and nonlinear~\cite{hermann1977nonlinear} systems. In his celebrated work, Kalman~\cite{kalman1960} derived the filter by exploiting the equivalence between the Gaussian estimation problem and the linear quadratic optimal control problem. This elegant duality has long served as a foundational concept in control theory.    

However, as modern application demands the estimation and control of increasingly high-dimensional, nonlinear dynamical systems, the classical approach becomes less feasible. Specifically, explicitly maintaining and propagating the error covariance matrix becomes computationally intractable due to its quadratic scaling with the state dimension. To address this, the Ensemble Kalman filter (EnKF) was developed by practitioners in geophysical sciences~\cite{calvello2025ensemble}. It was designed to handle models with massive state spaces---often on the order of millions to billions. 
The celebrated work of Evensen~\cite{evensen1994sequential} introduced the paradigm of using Monte-Carlo ensemble of state estimators to approximate this covariance matrix. Since then, the EnKF has become the standard approach for state estimation in high-dimensional, nonlinear dynamical systems, particularly within the domains of numerical weather prediction, oceanography, and geophysics.

In this paper, we are revisiting the development of EnKF and its variants from an optimal control perspective to establish a unified derivation framework in a discrete-time setting. This extends the initial idea from the author's previous work~\cite{kim2018LFPF}, which established a similar duality-based framework for continuous-time linear systems, into the discrete domain where practical EnKF algorithms should operate. A duality formalism for discrete time hidden Markov models with finite state and observation spaces appears in recent preprint~\cite{chang2026dual}.

\subsection{Variants of EnKF}\label{ssec:variants-of-EnKF}

The evolution of the EnKF has yielded several distinct methodologies to address the specific statistical and computational challenges of high-dimensional data assimilation. A brief overview of these primary variants is essential to understand the motivation behind our unified framework. Table~\ref{tb:comparison} summarizes necessary takeaways in this section. An excellent and comprehensive review of these methodologies can be found in~\cite{calvello2025ensemble} and the references therein.


\begin{table*}
	\centering
	\renewcommand{\arraystretch}{2}
	\caption{Summary table of Discrete-time EnKF variants} \label{tb:comparison}
	\small
	\begin{tabular}{p{0.22\textwidth}p{0.35\textwidth}p{0.35\textwidth}}
		{\bf EnKF Variants}& {\bf Analysis Mechanism} & {\bf Features and Benefits} \\
		\hline \hline
		Original Ensemble Kalman Filter (EnKF)~\cite{evensen1994sequential} & Applies full Kalman gain to the innovation with perturbed observations.
		& Injects copy of process and measurement noise. Robust to nonlinearity and non-Gaussianity.
		\\ \hline
		Ensemble Square Root Filter (EnSRF)~\cite{whitaker2002ensemble} & Applies a deterministically reduced gain to the anomalies. & Avoid sampling error due to the perturbed observations. \\ \hline
		Deterministic Ensemble Kal-man Filter (DEnKF)~\cite{sakov2008deterministic} & Fix the reduction rate in EnSRF to half. & Very efficient linear approximation of the EnSRF. \\
		\hline  
		Ensemble Adjustment Kalman Filter (EAKF)~\cite{anderson2001ensemble} & Left-multiplication of the anomaly matrix directly in the state space & Ensures exact covariance matching. High computational cost unless treated properly.
		\\ \hline
		Ensemble Transform Kalman Filter (ETKF)~\cite{bishop2001adaptive} & Right-multiplication of the anomaly matrix within the low-dimensional ensemble subspace.
		& Efficient for smaller number of ensembles compared to the state space dimension. 
		\\ \hline
		\hline
	\end{tabular}
\end{table*}

\newP{The Stochastic EnKF} The original formulation by Evensen~\cite{evensen1994sequential} utilizes Monte-Carlo ensemble of state model is simulated subjected to stochastic noise that represents the process noise. During the analysis step, the algorithm updates each individual ensemble member using classical Kalman gain. The covariance matrix is computed from the ensemble members. To ensure that the updated ensemble possesses the error covariance that is consistent with optimal linear filtering theory, the algorithm injects an independent copy of measurement noise into the innovation step for each ensemble member~\cite{burgers1998analysis}, and thus it referred to as the perturbed observation EnKF~\cite{van2020consistent}. Due to this artificially generated pseudo-random noises this approach introduces sampling errors that can degrade the filter's performance, particularly when the ensemble size is relatively small compared to the state dimension. However, it is noted that the stochastic algorithm is relatively more robust against non-Gaussian models~\cite{lei2010comparison}.

\newP{Ensemble Square Root Filter (EnSRF)} To eliminate the sampling procedure in the stochastic EnKF, a class of deterministic filters known collectively as Ensemble Square Root Filters was developed. However in this paper, we use the term EnSRF to refer to the formulation in~\cite{whitaker2002ensemble}. They noted that if the same unperturbed observations and the standard Kalman gain are used to update all ensemble members, the ensemble will systematically underestimate analysis-error covariances. To solve this, they updated the ensemble mean and the ensemble \emph{anomalies} (perturbations from the mean) separately, and adjusted the Kalman gain such that the covariance update equation exactly match in linear case without the use of random numbers. 

\newP{Deterministic Ensemble Kalman Filter (DEnKF)} This algorithm appears in~\cite{sakov2008deterministic} is formulated upon the recognition that, if the correction term is small, the adjustment required for the EnSRF is close to half. Hence the DEnKF updates the ensemble anomalies using exactly half the Kalman gain without computing a separate gain for the anomalies. It is claimed that the algorithm provides a linear approximation to the EnSRF with numerical efficiency.

\newP{Ensemble Adjustment Kalman Filter (EAKF)} Another deterministic approach is formulated in~\cite{anderson2001ensemble}. EAKF computes a linear operator that is applied to the ensemble estimate in the state space such that the updated ensemble exactly matches the theoretical analysis covariance. However, to compute such a huge operator, EAKF has to rely on sequential processing and localization techniques.

\newP{Ensemble Transform Kalman Filter (ETKF)} Since the linear operator of EAKF acts on the state space, explicit computation of the linear operator may not be feasible. ETKF introduced to solve this issue from an adaptive sampling perspective in~\cite{bishop2001adaptive}. It seeks a transformation matrix that operates on the ensemble subspace. This method is highly computationally efficient because it computes the transformation weights in the much smaller dimension of the ensemble size, rather than the massive state space dimension.

\medskip

While these variants—and many other hybrid approaches—were developed from different algebraic or statistical motivations, they often appear disconnected. In the following sections, we will demonstrate how these seemingly distinct algorithms naturally emerge as specific solutions to an optimal control formulation of the estimation problem.

The remainder of the paper proceeds as follows: In Sec~\ref{sec:math}, the classic duality for linear systems is introduced, and then our definition of the optimization problem for ensemble filtering algorithm is presented. The optimal solution is given in Sec.~\ref{sec:opt}. Sec.~\ref{sec:enkf} includes the connection to variety of EnKF family to this optimal control formulation is presented. All proofs appear in the Appendix.

\section{Mathematical Preliminaries}\label{sec:math}


In this paper, we consider the discrete time linear Gaussian dynamical system described by the following:
\begin{subequations}\label{eq:sys}
	\begin{align}
		X_{t+1} &= AX_t + \xi_t \label{eq:sys-a}\\
		Z_t &= HX_t + \zeta_t \label{eq:sys-b}
	\end{align}
\end{subequations}
where $\{X_t\in\Re^n:t=0,1,\ldots\}$ is the state process and $\{Z_t\in\Re^m:t=0,1,\ldots\}$ is the observation process. The initial condition is drawn from a Gaussian distribution $X_0\sim N(m_0,\Sigma_0)$. $A$ and $H$ are the system matrices of appropriate dimension, possibly time-dependent. The additive noise processes $\{\xi_t\in\Re^n:t=0,1,\ldots\}$ and $\{\zeta_t\in\Re^m:t=0,1,\ldots\}$ are mutually independent i.i.d.~Gaussian random variables with covariance $Q$ and $R$, respectively.

The goal of the estimation problem is to find the optimal estimator of the state process $X_T$ at some time $T$, given the observation up to time $T-1$, denoted by $Z_{0:T-1}:=\{Z_t:t=0,1,\ldots,T-1\}$. Here, ``optimal estimator'' denotes the one that minimizes the mean-squared error, which is essentially the conditional expectation $\E[X_T|Z_{0:T-1}]$.

\medskip

For the linear Gaussian problem~\eqref{eq:sys-a}-\eqref{eq:sys-b}, the posterior distribution is also a Gaussian $N(m_t,\Sigma_t)$ whose mean and covariance are given by the Kalman filter:
\begin{subequations}\label{eq:KF}
	\begin{align}
		m_{t+1} &= A\big(m_t + K_t(Z_t-Hm_t)\big) \label{eq:KF-mean}\\
		\Sigma_{t+1} &= A\Sigma_t A^\tp + Q - A\Sigma_t H^\tp (H\Sigma_t H^\tp + R)^{-1}H\Sigma_t A^\tp \label{eq:KF-var}
	\end{align}
\end{subequations}
where
\begin{equation} \label{eq:kt_def}
	K_t:=\Sigma_t H^\tp(H\Sigma_t H^\tp+R)^{-1}
\end{equation}
is the Kalman gain.

%

\subsection{Classical duality for linear problems}

Classical duality approach in discrete time seeks to construct a linear estimator that optimally estimates an arbitrary linear functional of the final state, denoted as $a^\tp X_T$, given the observation signal up to time $T-1$. Introducing the scalar projection $a^\tp$ for an arbitrary $a\in\Re^n$ is for notational simplicity. Since $a$ is arbitrary, one can later drop $a$ from the final linear expressions, yielding the exact same result as a direct vector(matrix)-based formulations.

Consider a causal estimator with the linear structure given as follows:
\begin{equation}\label{eq:estimator}
	S_T = b^\tp m_0 - \sum_{t=0}^{T-1} u_t^\tp Z_t
\end{equation}
It is parameterized with a deterministic vector $b\in\Re^n$ and the weight process for each incoming observation signal $u:=\{u_t\in\Re^m:t=0,1,\ldots,T-1\}$. 

The optimal parameters are obtained such that the following mean-squared error problem:
\begin{equation}\label{eq:opt-classic}
	\min_{b,u}\;\E\big[|a^\tp X_T-S_T|^2\big]
\end{equation}
subjected to the system equation~\eqref{eq:sys-a}-\eqref{eq:sys-b} and the estimator $S_T$ is also a random variable (due to $Z_t$) defined by~\eqref{eq:estimator}.

The duality-based derivation of Kalman filter utilizes the backward-in-time control process to convert the optimization problem~\eqref{eq:opt-classic} into a deterministic linear quadratic problem. Detailed computation appears in the Appendix~\ref{apdx:classic}. 

\subsection{Ensemble of estimators}

In order to elevate classic framework from a singular state estimator to an ensemble of particles, it is necessary to introduce independent statistical copies of the fundamental sources of randomness. Clearly, there are three of them: the initial condition $X_0$, the process noise $\xi_t$ and the measurement noise $\zeta_t$. Let $\tilde{X}_0\in\Re^n$ be an independent copy of the initial condition, and let $\tilde{\xi}_t\in\Re^n$ and $\tilde{\zeta}_t\in\Re^m$ be independent copies of the process and measurement noise, respectively. The linear estimator is subsequently expanded into a parameterized equation that incorporates these additional sources of randomness, scaled by some parameter vectors to be determined. The augmented estimator is defined by:
\begin{equation}\label{eq:estimator-aug}
	\tilde{S}_T = b^\tp m_0 - \sum_{t=0}^{T-1}u_t^\tp Z_t + c^\tp(\tilde{X}_0-m_0) + \sum_{t=0}^{T-1}v_t^\tp\tilde{\xi}_t + \sum_{t=0}^{T-1} w_t^\tp \tilde{\zeta}_t
\end{equation}
where $b\in \Re^n$, $c\in\Re^n$, $u=\{u_t\in \Re^m:t=0,1,\ldots,T-1\}$, $v=\{v_t\in \Re^n:t=0,1,\ldots,T-1\}$, $w=\{w_t\in \Re^m:t=0,1,\ldots,T-1\}$.

\medskip

\begin{remark}
	Note that instead of formulating the ensemble as finite $N$ particles---as is common in standard Ensemble filter literature---we adopt a mean-field perspective. In this framework, the variables $\tilde{X}_0$, $\tilde{\xi}_t$ and $\tilde{\zeta}_t$ are treated as random variables in their respective spaces ($\Re^n$ or $\Re^m$) representing a generic single ensemble member. Consequently, the ensemble statistics such as the mean and variance are computed by the marginals of the underlying probability measures rather than empirical estimates.
\end{remark}

\medskip
As in the classic duality framework~\eqref{eq:opt-classic}, the optimality of the estimator is defined in the minimum variance sense. To determine the ``optimal'' parameters for the ensemble estimator~\eqref{eq:estimator-aug}, we formulate the following bi-objective optimization problem:
\begin{itemize}
	\item[{\bf (o1)}] The ensemble mean, $\E[\tilde{S}_T|Z_{0:T-1}]$, is the minimum variance estimator of $a^\tp X_T$.
	\item[{\bf (o2)}] The ensemble second moment, $\E\big[\tilde{S}_T^2|Z_{0:T-1}\big]$, is the minimum variance estimator of $(a^\tp X_T)^2$.
\end{itemize}
In the following section, the optimal solution that solves both objective is presented.

\section{The Optimal Solution}\label{sec:opt}

The first objective boils down to the classic duality, and the optimal choice of the parameters $b$ and $\{u_t\}$ are obtained by the following proposition:

\medskip

\begin{proposition}\label{prop:b-u}
	Consider a state transition matrix defined by
\begin{equation*}
	\Phi_{T,t} := \prod_{s=t}^{T-1}\big(I - H^\tp K_s^\tp\big)A^\tp
\end{equation*}
Then the choice of parameters
\begin{subequations}
	\begin{align*}
		b &= \Phi_{T,0}a\\
		u_t&= -K_t^\tp A^\tp \Phi_{T,t+1}a
	\end{align*}
\end{subequations}
makes the estimator~\eqref{eq:estimator-aug} satisfy (o1) regardless of the choice of $c$, $v$ and $w$.
\end{proposition}
\begin{proof}
	$c$, $v_t$ and $w_t$ all vanish upon taking ensemble mean. Henceforth the estimator~\eqref{eq:estimator-aug} becomes identical to~\eqref{eq:estimator}, and the problem boils down to the classic duality. The derivation of the optimal control process appears in the Appendix~\ref{apdx:classic}.
\end{proof}

\medskip
The following theorem describes the optimal choice for the remaining parameters that simulatneously satisfies the optimization objectives (o1) and (o2).

\medskip

\begin{theorem}\label{prop:c-v-w}
	 Let $\gamma_1, \gamma_2 \in [0,1]$ be two hyperparameters to be chosen. Suppose $\Sigma_t\succ 0$ for all $t=0,1,\ldots,T-1$. Then there exists a matrix $C_t\in \Re^{n\times n}$ that satisfies
	\begin{equation}\label{eq:C-t}
		\begin{aligned}
			C_t^\tp \Sigma_t C_t &= \big(A-\frac{1+\gamma_2^2}{2}AK_tH\big)\Sigma_t\big(A-\frac{1+\gamma_2^2}{2}AK_tH\big)^\tp \\
			&\quad -\Big(\frac{1-\gamma_2}{2}\Big)^2AK_tH\Sigma_tH^\tp K_t^\tp A^\tp + (1-\gamma_1^2)Q
		\end{aligned}
	\end{equation}
	Then we choose
	\begin{subequations}
		\begin{align*}
			b &= \Phi_{T,0}a\\
			u_t&= -K_t^\tp A^\tp \Phi_{T,t+1}a\\
			c &= \Psi_{T,0} a\\
			v_t &= \gamma_1 \Psi_{T,t+1} a\\
			w_t &= \gamma_2 K_t^\tp A^\tp \Psi_{T,t+1} a
		\end{align*}
	\end{subequations}
	where $\Phi_{T,t}$ is defined as in Prop.~\ref{prop:b-u}, and 
	\begin{equation*}
		\Psi_{T,t} := \prod_{s=t}^{T-1}C_t
	\end{equation*}
	Then the augmented estimator
	\begin{align*}
		\tilde{S}_T^* &= a^\tp \Big(\Phi_{T,0}^\tp m_0 + \sum_{t=0}^{T-1}\Phi_{T,t+1}^\tp AK_t Z_t + \Psi_{T,0}^\tp (\tilde{X}_0-m_0) \\
		&\quad +  \gamma_1\sum_{t=0}^{T-1} \Psi_{T,t+1}^\tp \tilde{\xi}_t + \gamma_2\sum_{t=0}^{T-1}\Psi_{T,t+1}^\tp A K_t \tilde{\zeta}_t\Big)
	\end{align*}
	satisfies (o1) and (o2) simultaneously.
\end{theorem}
\begin{proof}
See Appendix~\ref{apdx:pf-prop-cvw}.
\end{proof}

\medskip

\begin{remark}
	The $\Phi$ and $\Psi$ are the state transition matrix of the closed loop system of the optimal control problem, and the optimal control $u$, $v$ and $w$ are of the feedback form of the backward-in-time dual processes. In continuous-time setting, the counterpart of $\Psi$ was available in explicit form~\cite{kim2018LFPF}:
	\[
	\frac{\ud\Psi}{\ud t} = \Big(-A^\tp +\frac{1+\gamma_2^2}{2}H^\tp K_t^\tp -\frac{1-\gamma_1^2}{2}\Sigma_t^{-1}Q_t\Big)\Psi
	\]
	However in our work, the $C_t$ matrix is given implicitly via quadratic matrix equation and can only be obtained explicitly under limited conditions. Nevertheless, the non-uniqueness opens up possibility for exploring various algorithmic approaches even under the same hyperparameters $\gamma_1$ and $\gamma_2$, as we outline in the following section.
\end{remark}

\medskip

\section{Connection to the Variety of EnKFs}\label{sec:enkf}

The recursive formulation for the ensemble filter obtained from Thm.~\ref{prop:c-v-w} is presented in the following proposition.

\medskip

\begin{theorem}\label{prop:recursive}
	For any given $a\in\Re^n$ and $T\ge 0$, the optimal estimator $\tilde{S}_T^* = a^\tp \tilde{X}_T$, where $\tilde{X}_T$ is the solution to
	\begin{equation}\label{eq:recursive}
		\tilde{X}_{t+1} = A\big(\tilde{m}_t+K_t(Z_t-H\tilde{m}_t+\gamma_2\tilde{\zeta}_t)\big) + C_t^\tp (\tilde{X}_t-\tilde{m}_t)+\gamma_1\tilde{\xi}_t
	\end{equation}
	with initial distribution $\tilde{X}_0\sim N(m_0,\Sigma_0)$, and $\tilde{m}_t = \E[\tilde{X}_t|Z_{0:t-1}]$ is the ensemble mean from the previous time step.
\end{theorem}
\begin{proof}
	It is a bit tedious but straightforward computation. See Appendix~\ref{apdx:pf-prop-rec}.
\end{proof}

\medskip

The recursive formula has three terms. The first term is the stochastic EnKF term, and the second term involves $C_t$ as the gain multiplied to the anomalies. Third term is the copy of the process noise. The process and measurement noises are scaled by the choice of parameters $\gamma_1$ and $\gamma_2$, respectively.

\subsection{Stochastic EnKF}

Let $\gamma_1 = \gamma_2 = 1$ to admit copies of both process noise and measurement noise. Then equation~\eqref{eq:C-t-alt} (an alternative form of~\eqref{eq:C-t}) becomes
\begin{equation*}
	C_t^\tp \Sigma_t C_t = \big(A-AK_tH\big)\Sigma_t\big(A-AK_tH\big)^\tp
\end{equation*}
and therefore the natural choice is $C_t = (I-H^\tp K_t^\tp) A^\tp$. 
The corresponding filter from Thm.~\ref{prop:recursive} is
\begin{equation*}
	\tilde{X}_{t+1} = A\big(\tilde{X}_t+K_t(Z_t-H\tilde{X}_t+\tilde{\zeta}_t)\big) + \tilde{\xi}_t
\end{equation*}
which recovers the stochastic EnKF~\cite{evensen1994sequential}.

\subsection{Eliminating perturbed observation}

The copy of the measurement noise $\tilde{\zeta}_t$ appears in~\eqref{eq:recursive} by the parameter $\gamma_2$. In order to avoid the injection of auxiliary noise to the observation, we choose $\gamma_2 = 0$ while we keep $\gamma_1=1$ the same as before. In this case, equation~\eqref{eq:C-t} becomes
\begin{equation}\label{eq:C-t-no-perturb}
	C_t^\tp \Sigma_t C_t = \big(A-\frac{1}{2}AK_tH\big)\Sigma_t\big(A-\frac{1}{2}AK_tH\big)^\tp - \frac{1}{4}AK_tH\Sigma_tH^\tp K_t^\tp A^\tp
\end{equation}
First, we drop the second term accepting a some margin of error. Then we have $C_t = \big(A-\frac{1}{2}AK_tH\big)^\tp$, which will make~\eqref{eq:recursive}
\begin{equation*}
	\tilde{X}_{t+1} = A\Big(\tilde{X}_t+K_t\big(Z_t-\frac{H\tilde{m}_t+H\tilde{X}_t}{2}\big)\Big) + \tilde{\xi}_t
\end{equation*}
This corresponds to the DEnKF suggested in~\cite{sakov2008deterministic}, where the authors also drop the term quadratic in $K_tH$. The error may be acceptable in the case when the observation error $R$ is relatively large, and thus the correction term becomes small.

Meanwhile, the original formulation of EnSRF~\cite{whitaker2002ensemble} further assumes $R$ is diagonal, and examine the observation individually. To connect this idea to our discussion, set $m=1$. Then $H\Sigma_t H$ and $R$ become scalar, and one can solve for $\alpha$ such that
an anzats $C_t = (A-\alpha AK_tH)^\tp$ satisfies~\eqref{eq:C-t-no-perturb}. It is straightforward that $\alpha$ is the solution to the quadratic equation
\begin{equation*}
	\frac{H\Sigma_t H^\tp}{H\Sigma_t H^\tp + R}\alpha^2 - 2\alpha + 1 = 0
\end{equation*}
which is exactly the reduction rate in~\cite{whitaker2002ensemble}.

In general, $C_t$ is formally obtained by 
\begin{equation*}
	C_t = \Sigma_t^{-1/2}L\Gamma_t^{1/2}
\end{equation*}
where $L$ is an arbitrary unitary matrix such that $LL^\tp = I$ and $\Gamma_t$ is the right-hand side of~\eqref{eq:C-t}. The goal of EAKF is to find $C_t$ directly using ensembles.  However, this may involve computing inverse and square root of $n\times n$ matrices. An approach to compute the inverse square root $\Sigma_t^{-1/2}$ is to use singular value decomposition on the ensemble anomalies: $\Sigma_t \approx \frac{1}{N-1} \delta_t\delta_t^\tp$ where $\delta_t$ is the $n\times N$ anomaly matrix, $N$ is the number of particles in the ensemble. By performing the singular value decomposition on the anomaly matrix $\delta_t=USV^\tp$, one can obtain the inverse square root by $\Sigma_t^{-1/2}\approx \sqrt{N-1}US^{-1}V^\tp$.  

Lastly, ETKF reformulate the matrix square-root problem into finding $N\times N$ matrix $W_t$ such that $\frac{1}{N-1}\delta_tW_tW_t^\tp\delta_t^\tp \approx \Gamma_t$, instead of finding $n\times n$ matrix $C_t$. The algorithm also changes such that $W_t$ is multiplied to the anomaly term from the right.

\subsection{Discussions on general cases}\label{sec:conclusion}

The freedom to choose the parameters $\gamma_1$ and $\gamma_2$ may reveal a hybrid configuration. For instance, one can choose small but non-zero $\gamma_2$ to benefit from the robustness of the stochastic formulation. Computing $C_t$ in this case can be done similar to DEnKF or EnSRF. A completely deterministic setting, $\gamma_1=\gamma_2=0$, yields the right-hand side of~\eqref{eq:C-t} be the discrete time Riccati equation. While we cannot find explicit formula for $C_t$ from the quadratic matrix equation, this way we can move covariance inflation from the process model into interaction term between particles which is proportional to the anomaly.

\section{Conclusion and Future outlook}

In this work, different variants of ensemble Kalman filter, both with or without perturbed observation, are revisited. Through the lens of duality, it is demonstrated that operational choices between seemingly distinct algorithms are, in fact, a choice of hyperparameters.

The non-uniqueness of the optimal solution leads to various deterministic ensemble algorithm even under the same choice of hyperparameters. This original formulation also opens up new possibilities for designing hybrid algorithms by tuning these hyperparameters $\gamma_1$ and $\gamma_2$. 
Previous studies have shown deterministic filters (where $\gamma$'s are nearly $0$) have smaller sampling error, but often fails in strongly nonlinear regimes. Conversely, stochastic noise injection performs better preserving nonlinearity and non-Gaussianity. Future theoretical efforts should explore how to navigate the $\gamma$ homotopy to manage simulation variance, finite $N$ effects or structural stability when applied to extreme systems, potentially establishing an algorithm to optimally choose $\gamma$'s.

Extending this framework to nonlinear, non-Gaussian settings will be the most crucial next step.  In the linear-Gaussian settings, matching the first two moments ((o1) and (o2)) is sufficient to establish an exact filter. However, in nonlinear settings, and exact duality-based filter would theoretically require matching all higher-order moments. Furthermore, achieving such exactness cannot rely on the deterministic weight processes utilized in this paper. Instead, a stochastic backward dual processes shall be introduced (as explored in continuous-time settings in\cite{kim2019duality,duality_jrnl_paper_II}). 

In theory, the EnKF variants showcased in this paper may become suboptimal when applied to nonlinear, non-Gaussian posteriors. However, the practical success of the EnKF family is largely due to their robust nonlinear extensions. In practice, the EnKF remains effective by simply propagating the ensemble members through the full nonlinear model update, yielding robust estimates nonetheless. Within the proposed framwork, replacing the linear observation matrix $H$ with a nonlinear function is straightforward. However, constructing adjoint for a nonlinear state model remains challenging and needs further investigation.



\section{Acknowledgments}
This work is supported by the Hongik University New Faculty Research Funds. The author has used AI assistant, specifically Gemini and DeepL, to translate portions of the text from the author's native language and to refine the English expressions in this manuscript. 

\bibliographystyle{IEEEtran}
\bibliography{_master_bib_jin,jin_papers}

\appendix

\section{Appendix}

\subsection{Derivation of Kalman filter using duality}\label{apdx:classic}

We loosely follow the formulation in~\cite[Ch.~7.5]{astrom1970}.
Consider the backward-in-time control process
\begin{equation}\label{eq:y-classic}
	y_t = A^\tp y_{t+1}+H^\tp u_t,\quad y_T = a
\end{equation}
for an arbitrary vector $a\in\Re^n$. Then
\begin{align*}
	a^\tp X_T &= y_T^\tp X_T\\
	&= y_0^\tp X_0 + \sum_{t=0}^{T-1} \big(y_{t+1}^\tp X_{t+1} - y_t^\tp X_t\big)\\
	&= y_0^\tp X_0 + \sum_{t=0}^{T-1} \big(y_{t+1}^\tp (AX_t + \xi_t) - (A^\tp y_{t+1}+H^\tp u_t)^\tp X_t\big)\\
	&= y_0^\tp X_0 + \sum_{t=0}^{T-1} \big(y_{t+1}^\tp \xi_t - u_t^\tp (Z_t-\zeta_t)\big)
\end{align*}
We set $b = y_0$ and then the optimization objective becomes
\begin{equation*}
	\E\big[|a^\tp X_T - S_T|^2\big] = y_0^\tp \Sigma_0 y_0 + \sum_{t=0}^{T-1} y_{t+1}^\tp Q y_{t+1} + u_t^\tp R u_t
\end{equation*}
which is a quadratic control objective function.

The corresponding optimal control is given by
\begin{equation*}
	u_t^\opt = -(H\Sigma_t H^\tp + R)^{-1} H\Sigma_t A^\tp y_{t+1} = -K_t^\tp A^\tp y_{t+1}
\end{equation*}
where
\begin{equation*}\label{eq:Ricc-discrete}
	\Sigma_{t+1} = A\Sigma_t A^\tp + Q - A\Sigma_t H^\tp (H\Sigma_t H^\tp + R)^{-1} H\Sigma_t A^\tp
\end{equation*}
Note that this is exactly the Riccati equation~\eqref{eq:KF-var}.

Now plug the optimal control into the dual process~\eqref{eq:y-classic} to define a state transition matrix for the closed loop system:
\begin{equation*}
	\Phi_{T,t} := \prod_{s=t}^{T-1}\big(I - H^\tp K_s^\tp\big)A^\tp
\end{equation*}
so that $y_t = \Phi_{T,t} a$. Now $S_T = a^\tp m_T$ where
\begin{align*}
	m_T &= \Phi_{T,0}^\tp m_0 + \sum_{t=0}^{T-1}\Phi_{T,t+1}^\tp AK_tZ_t\\
	&= A\big(m_{T-1} +K_{T-1}(Z_{T-1}-Hm_{T-1})\big)
\end{align*}
Thus we obtain the mean update equation~\eqref{eq:KF-mean} for the Kalman filter. \qed

\subsection{Proof of Theorem~\ref{prop:c-v-w}}\label{apdx:pf-prop-cvw}

Let
\[
\tilde{m}_T := b^\tp m_0 - \sum_{t=0}^{T-1}u_t^\tp Z_t = \E[\tilde{S}_T|Z_{0:T-1}]
\]
The estimator~\eqref{eq:estimator-aug} is therefore
\[
	\tilde{S}_T = \tilde{m}_T + c^\tp(\tilde{X}_0-m_0) + \sum_{t=0}^{T-1}v_t^\tp\tilde{\xi}_t + \sum_{t=0}^{T-1} w_t^\tp \tilde{\zeta}_t
\]
The optimization objective (o2) is expressed by
\begin{align*}
	\min_{b,u,c,v,w}\E\Big[\big(\E[\tilde{S}_T^2|Z_{0:T-1}]-(a^\tp X_T)^2\big)^2\Big]
\end{align*}
The first term inside the square is
\begin{equation*}
	\E[\tilde{S}_T^2|Z_{0:T-1}] = \tilde{m}_T^2 + c^\tp \Sigma_0 c + \sum_{t=0}^{T-1}v_t^\tp Q v_t + \sum_{t=0}^{T-1}w_t^\tp R w_t
\end{equation*}
Therefore
\begin{align*}
	\E\Big[\big(\E[&\tilde{S}_T^2|Z_{0:T-1}]-(a^\tp X_T)^2\big)^2\Big]\\
	&=\E\Big[\big(\tilde{m}_T^2 + c^\tp \Sigma_0 c + \sum_{t=0}^{T-1}v_t^\tp Q v_t + \sum_{t=0}^{T-1}w_t^\tp R w_t-(a^\tp X_T)^2\big)^2\Big]
\end{align*}
Note that $c^\tp \Sigma_0 c + \sum_{t=0}^{T-1}v_t^\tp Q v_t + \sum_{t=0}^{T-1}w_t^\tp R w_t$ is deterministic. The objective becomes
\begin{align*}
\E\Big[\big(\E[&\tilde{S}_T^2|Z_{0:T-1}]-(a^\tp X_T)^2\big)^2\Big] = \var\Big[\tilde{m}_T^2 - (a^\tp X_T)^2\Big] \\
&+\Big(c^\tp \Sigma_0 c + \sum_{t=0}^{T-1}v_t^\tp Q v_t + \sum_{t=0}^{T-1}w_t^\tp R w_t-\E\big[(a^\tp X_T)^2-\tilde{m}_T^2\big]\Big)^2
\end{align*}
By exploiting orthogonality principle of the first moment estimator, the variance term is also minimized when (o1) is solved (See Appendix A in~\cite{kim2018LFPF}). In this case, the conditional variance $\E\big[(a^\tp X_T)^2-\tilde{m}_T^2\big] = a^\tp \Sigma_T a$ is known. Now can we perform the matching for the second moment. The second term is minimized when
\[
c^\tp \Sigma_0 c + \sum_{t=0}^{T-1}v_t^\tp Qv_t + \sum_{t=0}^{T-1}w_t^\tp R w_t - a^\tp \Sigma_T a = 0
\]
where $\Sigma_T$ is the solution to~\eqref{eq:KF-var}, if achievable.

Indeed, the zero is achievable by considering a backward-in-time process similar to the classical case:
\[
\eta_t = C_t\eta_{t+1},\quad \eta_T = a
\]
such that
\[
\Delta_t:=\eta_{t+1}^\tp \Sigma_{t+1} \eta_{t+1} - \eta_t^\tp\Sigma_t\eta_t = v_t^\tp Q v_t + w_t^\tp R w_t
\]
so that we can pick $c = \eta_0$. Now looking at the Riccati equation~\eqref{eq:KF-var}, we choose a solution to have the form
\[
v_t = \gamma_1 \eta_{t+1},\quad w_t = \gamma_2 K_t^\tp A^\tp \eta_{t+1}
\]
which yields
\begin{align*}
	\Delta_t &=\eta_{t+1}^\tp \Sigma_{t+1} \eta_{t+1} - \eta_t^\tp\Sigma_t\eta_t \\
	&=\eta_{t+1}^\tp\Sigma_{t+1}\eta_{t+1} - \eta_{t+1}^\tp C_t^\tp\Sigma_t C_t\eta_{t+1}\\
	&=\gamma_1^2 \eta_{t+1}^\tp Q \eta_{t+1} + \gamma_2^2\eta_{t+1}^\tp AK_tRK_t^\tp A^\tp \eta_{t+1}
\end{align*}
Since $\eta_{t+1}$ appears in the same from on both sides, we eliminate it and we use~\eqref{eq:kt_def} to simplify the equation to obtain the equation for $C_t$ as Eq.~\eqref{eq:C-t}.

From the Joseph form of the Kalman update for $\Sigma_{t+1}$, Eq.~\eqref{eq:C-t} can also be written as
\begin{equation}\label{eq:C-t-alt}
	\begin{aligned}
		C_t^\tp\Sigma_t C_t &= (A-AK_tH)\Sigma_t(A-AK_tH)^\tp \\
		&\quad + (1-\gamma_1^2)Q + (1-\gamma_2^2)AK_t RK_t^\tp A^\tp
	\end{aligned}
\end{equation}
Since the right-hand side is positive semidefinite for any $\gamma_1,\gamma_2\in[0,1]$, $C_t$ exists.

Finally, $\Psi_{T,t}$ is the state transition matrix for $\eta_t$ so that $\eta_t = \Psi_{T,t}a$.

\medskip

The equation of the estimator is exact because of the Gaussian assumption. \qed

\subsection{Proof of Theorem~\ref{prop:recursive}}\label{apdx:pf-prop-rec}
We start from defining $\tilde{S}_T = a^\tp \tilde{X}_T$ where
\begin{align*}
	\tilde{X}_T &= \Phi_{T,0}^\tp m_0 + \sum_{t=0}^{T-1}\Phi_{T,t+1}^\tp AK_t Z_t + \Psi_{T,0}^\tp (\tilde{X}_0-m_0) \\
	&\quad +  \gamma_1\sum_{t=0}^{T-1} \Psi_{T,t+1}^\tp \tilde{\xi}_t + \gamma_2\sum_{t=0}^{T-1}\Psi_{T,t+1}A K_t \tilde{\zeta}_t\\
	&= AK_{T-1}Z_{T-1} + \gamma_1 \tilde{\xi}_{T-1} + \gamma_2 AK_{T-1}\tilde{\zeta}_{T-1} \\
	&\quad + A(I-K_{T-1}H)\Big(\Phi_{T-1,0}^\tp m_0 + \sum_{t=0}^{T-2}\Phi_{T-1,t+1}^\tp AK_t Z_t\Big) \\
	&\quad+ C_{T-1}^\tp\Big(\Psi_{T-1,0}^\tp (\tilde{X}_0-m_0) +  \gamma_1\sum_{t=0}^{T-2} \Psi_{T-1,t+1}^\tp \tilde{\xi}_t \\
	&\qquad\qquad+ \gamma_2\sum_{t=0}^{T-2}\Psi_{T-1,t+1}A K_t \tilde{\zeta}_t\Big) \label{eq:mf-equation}\\
	&= AK_{T-1}Z_{T-1} + \gamma_1 \tilde{\xi}_{T-1} + \gamma_2 AK_{T-1}\tilde{\zeta}_{T-1} \\
	&\quad + A(I-K_{T-1}H)\tilde{m}_{T-1}+ C_{T-1}^\tp\Big(\tilde{X}_{T-1}-\tilde{m}_{T-1}\Big)
\end{align*}
The claim follows by substituting $T-1$ with $t$. \qed
\end{document}